\documentclass[12pt]{article}
\usepackage[latin1]{inputenc}

\usepackage{amssymb}
\usepackage[superscript]{cite}
\usepackage{enumerate}
\usepackage{amsmath}
\usepackage{amsthm}
\usepackage{tabularx}
\usepackage{longtable}
\usepackage{graphicx}
\usepackage{verbatim}
\usepackage{schemata}
\usepackage{bm}
\usepackage{color}
\usepackage{soul}
\usepackage[T1]{fontenc}

\newtheorem{theorem}{Theorem}

\newtheorem{corollary}[theorem]{Corollary}
\newtheorem{definition}[theorem]{Definition}
\newtheorem{example}[theorem]{Example}
\newtheorem{lemma}[theorem]{Lemma}

\newtheorem{remark}[theorem]{Remark}

\numberwithin{equation}{section} \numberwithin{theorem}{section}

\def\IR{\mathbb{R}}

\title{Preservation of some stochastic orders by distortion functions with application to coherent systems with exchangeable components}

\vspace{1cm}
\author{Antonio Arriaza\footnote{Universidad de C\'{a}diz, C\'{a}diz, Spain.}\,\,\,\thanks{Correspondence to: Antonio Arriaza, Facultad
        de Ciencias, Universidad de C\'{a}diz, Spain.\newline E-mail: \textit{antoniojesus.arriaza@uca.es}. Telephone number: $+$34 956 012 775.}\,\, and  \ Miguel A. Sordo$^{*}$ 
 }

\date{}

\begin{document}
    \maketitle

\begin{abstract}
  The preservation of stochastic orders by distortion functions has become a topic of increasing interest in the reliability analysis of coherent systems. The reason of this interest is that the reliability function of a coherent system with identically distributed components can be represented as a distortion function of the common reliability function of the components. In this framework, we study the preservation of the excess wealth order, the total time on test transform order, the decreasing mean residual live order and the quantile mean inactivity time  order by distortion functions.  The results are applied to study the preservation of these stochastic orders under the formation of coherent systems with exchangeable components.
  
\quad

\end{abstract}

\newpage

\section{Introduction}
This paper is concerned with the preservation of several stochastic orders useful in reliability theory by distortion functions. Unlike other preservation results in the literature, the relevant distortion functions considered in this paper  are starshaped and  antistarshaped,  two classes of functions that contain, respectively, the classes of convex and concave distortion functions. Since the reliability  function of a coherent system with identically distributed (ID) components can be represented as a distortion function of the common reliability function of the components, our results are applied to study the preservation of these stochastic orders under the formation of coherent systems.  

Let $X$ and $Y$ be two non-negative absolutely continuous random variables with finite means $\mu_X$ and $\mu_Y$ and distribution functions $F$ and $G,$ respectively. Let $f$ and $g$ be the corresponding density functions and $\bar{F}=1-F$ and $\bar{G}=1-G,$ respectively, the  reliability (or survival) functions. The quantile function of $X$ is defined by
\[
F^{-1}(p)=\inf\{x\in \mathbb R | F(x)\ge p\}, \text{ for all }p\in (0,1).
\] 
The quantile function of $Y$ is analogously defined. We recall the definitions of the stochastic orders considered in this paper. 

\begin{definition}\label{fourorders}
Under the above assumptions, we say that:
	\begin{itemize}
	\item[a)]$X$ is smaller than $Y$ in the total time on test transfom order, denoted by $X\leq _{ttt}Y,$ if 
	\begin{equation*}
	\int_0^{F^{-1}\left( p\right)}\bar{F}(x)dx \leq
	\int_ 0^{G^{-1}\left( p\right)}\bar{	G}(x)dx, \
	\text{ for all }p\in \left( 0,1\right)
	.
	\end{equation*}
		\item[b)] $X$ is smaller than $Y$ in the excess wealth order, denoted by $X\leq _{ew}Y,$ if 
		\begin{equation*}
		\int_{F^{-1}\left( p\right) }^{\infty }\bar{F}(x)dx \leq
		\int_{G^{-1}\left( p\right) }^{\infty }\bar{	G}(x)dx, \
		\text{ for all }p\in \left( 0,1\right)
		.
		\end{equation*}
		\item[c)]$X$ is smaller than $Y$ in the decreasing mean residual life order, denoted by $X\leq_{dmrl}Y$, if		
	\begin{equation*}
		\frac{\int_{G^{-1}\left( p\right) }^{\infty }\bar{G}(x)dx}{%
			\int_{F^{-1}\left( p\right) }^{\infty }\bar{F}(x)dx}\text{ increases in 
		}p\in \left(0,1\right).  \label{defdmrl}
	\end{equation*}
		\item[d)]$X$ is smaller than $Y$ in the the quantile mean inactivity time order, denoted by $X\leq_{qmit}Y$, if
		\begin{equation*}
				\frac{\int_{0}^{F^{-1}\left( p\right) }F(x)dx}{\int_{0}^{G^{-1}\left(
				p\right) }G(x)dx}
		\text{ decreases in }p \in (0,1).  \label{def2}
		\end{equation*}
		\end{itemize}
	\end{definition}

A classical book on the topic of stochastic orders is Shaked and Shanthikumar\cite{SS07}. Some applications in reliability theory of the total time on test transform  order and the excess wealth order can be found in Kochar et al. \cite{KLS2002}, Li and Chen \cite{LC2004}, Belzunce et al \cite{BMRS2016,BMRS2017} and Sordo and Psarrakos \cite{SP2017}.  Reliability applications of the decreasing mean residual life  order and the quantile mean inactivity time order, can be found in Barlow and Proschan\cite{BP75}, Kochar and Wiens\cite{KocharWiens1987}, Arriaza et al,\cite{ASS17,ABM19} and Kayid et al\cite{Kayid2018}.

A distortion function is a left continuous and increasing function $h:\left[ 0,1\right] \rightarrow \left[ 0,1\right] ,$ such that satisfies $%
h(0)=0$ and $h(1)=1.$ We denote by $\mathfrak{D}$ the class of all distortion functions. Given the survival function $\bar{F}$ of a random variable $X,$ the transformation $\bar{F}_{h}(x)=h\left( \bar{F}%
\left( x\right) \right)$ defines a new
survival function associated to some random variable $X_{h},$ which is considered as 
the distorted random variable induced from $X$ by $h.$ The distribution function $F_h=1-\bar{F}_{h}$ of $X_h$ satisfies $F_h(x)=h^{*}(F(x))$, where  $h^{*}(p)=1-h(1-p)$ is again a distortion function called the dual distortion function associated to $h$. Distorted distributions appear in reliability theory in the study of coherent systems. In general, the lifetime of a coherent system is described by a non-negative random variable $T$ that can be expressed as 
$$T=_{st} \phi(X_1,X_2,\ldots,X_n)$$ 
\noindent where $\phi:\IR^{+}\times\overset{(n)}{\cdots}\times  \IR^{+}\longrightarrow \IR^{+}$ is the structure function associated to the system and $X_i, \ i=1,...,n,$ is the lifetime of the $i$th component of the system (for definitions and basic properties of coherent systems see Barlow and Proschan\cite{BP75}).  
The dependence structure of the components of a system can be modelled by a copula function, which  is an $n$-dimensional distribution function with uniform marginals over the interval $(0,1)$. If the components are ID, the joint reliability function of the random vector $(X_1,\dots,X_n)$ can be expressed as
$$\Pr(X_1>x_1,\dots, X_n>x_n)=C(\bar F(x_1),\dots,\bar F(x_n)),$$
\noindent where $\bar F(x_i)=\Pr(X_i>x_i)$ is the common reliability function of the components and $C$ is  the survival copula associated to the random vector $(X_1,\dots,X_n)$. Then, the system reliability function $\bar F_T$ can be written as
\begin{equation}\label{hfbar}
\bar F_T(t)=h(\bar F(t)),
\end{equation} 
where $h(\cdot)$ is a distortion function which depends on the structure of the system and on the survival copula $C$ (see, e.g., Navarro et al\cite{NASS16}, Mizula and Navarro\cite{MN17}, Navarro\cite{N18}, and Navarro and Rychlik\cite{NavRych}). Analogously, the joint distribution function of $(X_1,\dots,X_n)$ can be expressed as
$$\Pr(X_1\leq x_1,\dots, X_n\leq x_n)=\hat{C}(F(x_1),\dots,F(x_n)),$$

\noindent with $\hat{C}$ the distributional copula associated to the random vector $(X_1,\dots,X_n)$ and $F(x_i)=\Pr(X_i\leq x_i)$ is the common distribution function of the components. The system distribution function $F_T=1-\bar F_T$ can be written as
\begin{equation}\label{hstarf}
F_T(t)=h^{*}(F(t)),
\end{equation} 
where $h^{*}(\cdot)$ is the dual distortion of $h(\cdot)$.

In reliability, the study of the preservation of stochastic orders under the formation of coherent systems is an important topic that has attracted
increasing attention  (see, for example, Navarro et al,\cite{NASS13, NASS14} Navarro and Gomis\cite{NG16}, Navarro and del Aguila\cite{NA17}, Arriaza et al\cite{ANS18},  Navarro and Sordo\cite{NS2018}, Navarro et al\cite{NAS18}, Li and Li\cite{LiLi2018} and Navarro and Cal\`{i}\cite{NC2019}). In the case of ID components, this is equivalent, in view of (\ref{hfbar}) and (\ref{hstarf}), to study the preservation of stochastic orders under distortion functions. Given two random variables $X$ and $Y$, and certain stochastic order $\leq_{(*)}$, the aim  is to find the largest subset $\mathfrak{F}\subseteq \mathfrak{D}$ such that:
\begin{equation*}
\mbox{if}\,\,X\leq_{(*)}Y \Longrightarrow X_{h}\leq_{(*)} Y_{h}\,\, \mbox{for all}\,\, h\in \mathfrak{F},
\end{equation*} 
\noindent where $X_h$ and $Y_{h}$ represent the distorted random variables by the function $h$ from $X$ and $Y$, respectively. For example, it is well-known that the total time on test transform order is preserved by convex and strictly increasing distortion functions (Li and Shaked\cite{LS2007}) and that the excess wealth order is preserved by concave and strictly increasing ones (Navarro et al\cite{NASS13}). 
A reasonable question is to ask whether these preservation results can be extended to  more general classes of distortion functions and whether these classes are useful for modelling purposes. To address this question, we consider in this paper the following classes of distortion functions.  
\begin{definition}
	Given a distortion function $h:[0,1] \rightarrow [0,1]$, we say that $h$ is starshaped (resp. antistarshaped) if $h(p)/p$ is increasing (resp. decreasing) for all $p\in(0,1].$
\end{definition}
We prove in Section 2 that the orders $\le_{ttt}$ and $\le_{ew}$ are preserved, respectively, by starshaped and antistarshaped distortion functions.  Since convex (resp. concave) distortion functions are starshaped (resp. antistarshaped), our results extend previous results in the literature. We also formulate in Section 2 preservation results in terms of the orders $\le_{dmrl}$ and $\le_{qmit}$. While these orders are not, in general, preserved by distortion functions (this is shown by counterexamples) there are cases in which they do. Specifically, the \textit{dmrl} order is preserved by antistarshaped distortion functions and the \textit{qmit} order is preserved by any distortion function $h$ such that its dual $h^*$ is antistarshaped.  

The results are applied in Section 3 to study the preservation of the above stochastic orders under the formation of coherent systems with  exchangeable components. A coherent system have exchangeable components if the joint distribution of the random vector of the component lifetimes is permutation invariant or, equivalently, if the components of the system are ID and the dependence among them is modelled by a symmetric copula (see Theorem 2.7.4 in Nelsen\cite{Nelsen06}). An $n$-dimensional copula $C:[0,1]\times\overset{(n)}{\cdots}\times[0,1]\rightarrow [0,1]$ is symmetric, if
$$C(u_1,u_2,\ldots,u_n)=C(u_{\pi(1)},u_{\pi(2)},\ldots,u_{\pi(n)})$$
	\noindent for any permutation $\pi:\{1,2,\ldots,n\}\rightarrow\{1,2,\ldots,n\}$.
The exchangeability of the components is a reasonable assumption when the system is formed by identical units and the failure of one of them affects equally the reliability of the remaining components (see Navarro and Rychlik\cite{NavarroRychlik2007}, Zhengcheng\cite{Zhengcheng2010} and Tavangar\cite{Tavangar2014}). As in Navarro et al\cite{NASS16}, the aim is to compare the lifetimes of two coherent systems $T_1=_{st} \phi(X_1,X_2,\ldots,X_n)$  and $T_2=_{st} \phi(Y_1,Y_2,\ldots,Y_n)$ with exchangeable components when $X_1 \le_{*}Y_1,$  where $\le_{*} $ is one of the orders under study. We illustrate the usefulness of our results by means of several examples of systems with exchangeable components where the reliability function of the system is modelled using starshaped (resp. antistarshaped) functions which are not convex (resp. concave).  Section 4 contains conclusions.

Throughout this paper, increasing means non-decreasing and decreasing means non-increasing.

\section{The main results}
\label{main}
\subsection{Preservation of the \textit{ttt} order  and the excess wealth order}
The following lemma will be used to prove the main results. For the proof, see Lemma 7.1(a) (and the remark below the lemma) in Chapter 4 of Barlow and Proschan\cite{BP81}.
\begin{lemma} \label{lBP}
Let $W$ be a measure on the interval $(a,b)$ and let $u$ be a non-negative function defined on 	$(a,b)$. If $\int_t^b dW(x)\ge 0$ for all $t\in (a,b)$ and if $u$ is increasing,  then  $\int_t^b u(x)dW(x)\ge 0$ for all $t\in (a,b).$
\end{lemma}

Next we show that the orders $\leq_{ttt}$ and $\leq_{ew}$ are preserved by starshaped and antistarshaped distortion functions, respectively. 
\begin{theorem} \label{THttt}
	Let $X$ and $Y$ be two non-negative continuous random variables with
	strictly increasing distribution functions $F$ and $G$. If $X\leq _{ttt}Y$,
	then $X_{h}\leq _{ttt}Y_{h}$ for all starshaped distortion function $h.$
\end{theorem}
\begin{proof}
	Let $h$ be a starshaped distortion function. Then, $h$ is
	strictly increasing except possibly where it is $0$ (otherwise, if $h(t)=c>0$
	on $\left( t_{1},t_{2}\right) ,$ $0<t_{1}<t_{2}\leq 1,$ then $h(t)/t$
	decreases on $\left( t_{1},t_{2}\right) $). Let $0\leq t_{1}<1$ be such that 
	$h(t)=0$ for $0\leq t\leq t_{1}$ and $h(t)>0$ for $t_{1}<t\leq 1.$  Assume
	that%
	\[
	\int_{0}^{F^{-1}(p)}\overline{F}(x)dx\leq \int_{0}^{G^{-1}(p)}\overline{G}%
	(x)dx,\text{ }p\in \left( 0,1\right) ,
	\]%
	or, equivalently, that 
	\[
	\int_{p}^{1}td\left[ \overline{F}^{-1}(t)-\overline{G}^{-1}(t)\right] \geq
	0,\ \ p\in \left( 0,1\right) .
	\]%
	Since $h(t)/t$ is increasing in $t\in $ $\left( 0,1\right] ,$ using Lemma \ref{lBP}, we have 
	\[
	\int_{p}^{1}h(t)d\left[ \overline{F}^{-1}(t)-\overline{G}^{-1}(t)\right]
	\geq 0,\ \ p\in \left( 0,1\right) .
	\]%
	Equivalently, we can write%
	\[
	\int_{p}^{1}h(t)d\left[ \overline{F}^{-1}(t)-\overline{G}^{-1}(t)\right]
	\geq 0,\ \ p\in \left( t_{1},1\right) .\ \ 
	\]
	The change of variable $h(t)=x$ yields to%
	\[
	\int_{h(p)}^{1}xd\overline{F}^{-1}(h^{-1}(x))\geq \int_{h(p)}^{1}xd\overline{%
		G}^{-1}(h^{-1}(x)),\ \ p\in \left( t_{1},1\right) .
	\]%
	Since $h(t)$ is strictly increasing on $\left( t_{1},1\right) $ with range $%
\left( 0,1\right) ,$ 
	this is equivalent to%
	\[
	\int_{p}^{1}xd\overline{F}^{-1}(h^{-1}(x))\geq \int_{p}^{1}xd\overline{G}%
	^{-1}(h^{-1}(x)),\ \ p\in \left( 0,1\right) .
	\]%
	Given $x\in \left( 0,1\right) ,$ there exists a unique $t$ such that $x=h\left( 
	\overline{F}(t)\right) .$ Making again the corresponding change of variable we
	have 
	\[
	\int_{0}^{F^{-1}(h^{-1}(p))}h\left( \overline{F}(t)\right) dt\leq
	\int_{0}^{G^{-1}(h^{-1}(p))}h\left( \overline{G}(t)\right) dt,\text{ }p\in
	\left( 0,1\right) ,
	\]%
	which is the same as%
	\[
	\int_{0}^{F_{h}^{-1}(p)}\overline{F}_{h}(x)dx\leq \int_{0}^{G_{h}^{-1}(p)}%
	\overline{G}_{h}(x)dx,\text{ }p\in \left( 0,1\right) ,
	\]%
	where $\overline{F}_{h}(x)=h\left( \overline{F}(x)\right) $ and $\overline{G}%
	_{h}(x)=h\left( \overline{G}(x)\right) $ are the reliability functions of
	the random variables $X_{h}$ and $Y_{h},$ respectively. This ends the proof.
\end{proof}

The proof of the following result is similar to the proof of Theorem \ref{THttt} and therefore it is omitted. In this case, the distortion function must be strictly increasing.

\begin{theorem}
\label{THew}
	Let $X$ and $Y$ be two non-negative continuous random variables with
	strictly increasing distribution functions $F$ and $G$. If $X\leq _{ew}Y$,
	then $X_{h}\leq _{ew}Y_{h}$ for all antistarshaped strictly increasing distortion function $h.$
\end{theorem}

\subsection{Preservation of the \textit{dmrl} order and the \textit{qmit} order} \label{sect2}
The  \textit{dmrl} order is related to the convex transform order (see Section 4.B in Shaked and Shanthikumar\cite{SS07}). Given two non-negative random variables $X$ and $Y$, with distribution functions $F$ and $G$, respectively, $X$ is said to be smaller than $Y$ in the convex transform order (denoted by $X\leq_{c}Y$) if $G^{-1}F(x)$ is convex for all $x\geq 0.$   Since $G^{-1}F(x)=G_h^{-1}F_h(x)$ for any distortion function $h$, it is obvious that the convex transform order is preserved by any distortion function. By considering 
\begin{equation}
\label{imply}
X\leq_{c}Y \Longrightarrow X\leq_{dmrl} Y,
\end{equation}
a reasonable question is whether the order \textit{dmrl} is also preserved by any distortion function. Unfortunately, this is not always the case,  as the counterexample given in Example \ref{ce02} below shows. First, we state two technical lemmas which will be used in the sequel. The first one is obtained by differentiation of (\ref{defdmrl}).

\begin{lemma}
	\label{L3}Let $X$ and $Y$ be two absolutely continuous random variables with
	respective distribution functions $F$ and $G$ and density functions $f$ and $%
	g,$ respectively. Then, $X\leq _{dmrl}Y$ if and only if%
	\begin{equation}
	\int_{p}^{1}\frac{\left( 1-t\right) dF^{-1}(t)}{g\left( G^{-1}\left(
		p\right) \right) }\leq \int_{p}^{1}\frac{\left( 1-t\right) dG^{-1}(t)}{%
		f\left( F^{-1}\left( p\right) \right) }\text{ for all }p\in \left(
	0,1\right) .  \label{equivdmrl}
	\end{equation}
\end{lemma}

\begin{lemma}
	\label{L4}Let $X$ and $Y$ be two absolutely continuous random variables with
	 strictly increasing distribution functions $F$ and $G,$ respectively. Let $f$ and $%
	g$ be the respective density functions. Then, $X\leq _{dmrl}Y$ if and only if%
	\begin{equation}
	\int_{q}^{1}\frac{\left( 1-t\right) dF^{-1}(t)}{g\left( G^{-1}\left(
		p\right) \right) }\leq \int_{q}^{1}\frac{\left( 1-t\right) dG^{-1}(t)}{%
		f\left( F^{-1}\left( p\right) \right) },\text{ for all }0<p\leq q<1.
	\label{ineqdmrl}
	\end{equation}
\end{lemma}

\begin{proof}
	From Lemma \ref{L3}, $X\leq _{dmrl}Y$ if and only if (\ref{equivdmrl}) holds
	or, equivalently,%
	\begin{equation}
	\frac{f\left( F^{-1}\left( p\right) \right) }{g\left( G^{-1}\left( p\right)
		\right) }\leq \frac{\int_{p}^{1}\left( 1-t\right) dG^{-1}(t)}{%
		\int_{p}^{1}\left( 1-t\right) dF^{-1}(t)},\ \text{for all }p\in \left(
	0,1\right) .  \label{1}
	\end{equation}%
	On the other hand, \eqref{1} is equivalent to%
	\begin{equation}
		\frac{\int_{p}^{1}\left( 1-t\right) dG^{-1}(t)}{\int_{p}^{1}\left(
		1-t\right) dF^{-1}(t)}\ \text{increases in }p\in \left[ 0,1\right] .
	\label{2}
	\end{equation}%
	Combining (\ref{1}) and (\ref{2}) we see that $X\leq _{dmrl}Y$ if and only
	if (\ref{ineqdmrl}) holds.
\end{proof}
The following counterexample shows that the order \textit{dmrl} is not preserved by general distortion functions.

\begin{example}
	\label{ce02}
	 To construct the counterexample, we need two non-negative random variables such that $X\leq_{dmrl}Y$ but $X\nleq_{c}Y$ (otherwise, it follows from the preservation property of the convex transform order and \eqref{imply} that $X_h\leq_{dmrl}Y_h$ for any distortion $h$). Note that given two absolutely continuous random variables $X$ and $Y$, with respective distribution functions $F$ and $G$ and density functions $f$ and $g$ respectively,  then $X\le_cY$ if and only if 
	 $$\frac{f(F^{-1}(p))}{g(G^{-1}(p))} \text{ is increasing for all } p\in (0,1).$$
	  Let $X$ and $Y$ be two non-negative random variables with respective quantile functions given by $F^{-1}(p)=\frac{17}{8}p-\frac{1}{2}p^2$ and $G^{-1}(p)=\ln(\frac{15}{8}+p)$, for $p \in (0,1)$. We first note that  $X\nleq_{c}Y$ because
	  $$\frac{f(F^{-1}(p))}{g(G^{-1}(p))}=\frac{1}{(17/8-p)(15/8+p)}$$
	   is decreasing in the interval $(0,0.125)$ and increasing in $(0.125,1)$, see Figure \ref{notconv02} (left). Now we compute the integral
		\begin{equation}
	\label{dmrlIntEqu}
	I(p)=\displaystyle \int_{p}^{1} (1-t)\left [\frac{1}{g(G^{-1}(t))}-\frac{f(F^{-1}(p))}{g(G^{-1}(p))}\frac{1}{f(F^{-1}(t))}\right ]dt, \ p\in (0,1).
	\end{equation}
It is shown in Figure \ref{notconv02} (right) that $I(p)\geq 0,\,\,\mbox{for all}\,\,p\in (0,1).$ It follows from Lemma \ref{L3} that $X\leq_{dmrl}Y.$
\begin{figure}
		\begin{center}
			\includegraphics*[scale=0.65]{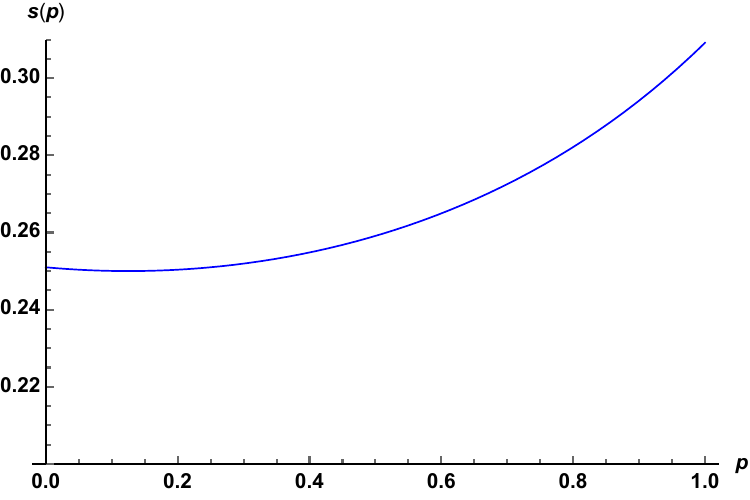}
			\includegraphics*[scale=0.65]{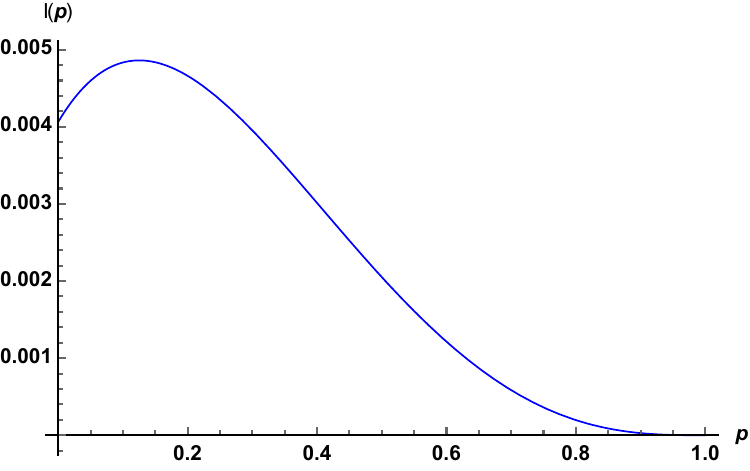}
			\caption{Plot of the function $s(p)=f(F^{-1}(p))/g(G^{-1}(p))$ in Example \ref{ce02} (left).  Plot of the integral $I(p)$ defined in \eqref{dmrlIntEqu} for all $p\in (0,1)$ (right).} \label{notconv02}
		\end{center}
	\end{figure}
Now consider the distortion function $h(p)=p^5$. The quantile functions $F^{-1}_{h}(p)$ and $G^{-1}_{h}(p)$ of $X_h$ and $Y_h$, respectively, are given by 
	$$F^{-1}_{h}(p)=F^{-1}(1-h^{-1}(1-p))=17/8 (1 - (1 - p)^{1/5}) - 1/2 [1 - (1 - p)^{1/5}]^2$$
	\noindent	and
	$$G^{-1}_{h}(p)=G^{-1}(1-h^{-1}(1-p))=\ln[23/8 - (1 - p)^{1/5}],$$
for $p\in (0,1)$. By taking derivatives we have 
	\begin{equation}
	\label{dmrlex01}
	f_h(F_h^{-1}(p))=\frac{5(1-p)^{4/5}}{\frac{9}{8}+(1-p)^{1/5}}
	\end{equation}
	\noindent and
	\begin{equation}
	\label{dmrlex02}
	g_h(G_h^{-1}(p))=5(1-p)^{4/5}(23/8-(1-p)^{1/5}),
	\end{equation}
for $p\in (0,1)$. 	
	Denote
	\begin{equation}
\label{dmrlIntEqu_h}
I_h(p)=\displaystyle \int_{p}^{1} (1-t)\left [\frac{1}{g_h(G_h^{-1}(t))}-\frac{f_h(F_h^{-1}(p))}{g_h(G_h^{-1}(p))}\frac{1}{f_h(F_h^{-1}(t))}\right ]dt, \ p\in (0,1).
\end{equation}	
Replacing  \eqref{dmrlex01} and \eqref{dmrlex02} into \eqref{dmrlIntEqu_h} we observe  that $I_h(p) \le 0$ for $p\in[0,0.02632]$ (see Figure \ref{notconv03}) and therefore $X_h\nleq_{dmrl}Y_h.$

	\begin{figure}
		\begin{center}
			\includegraphics*[scale=0.65]{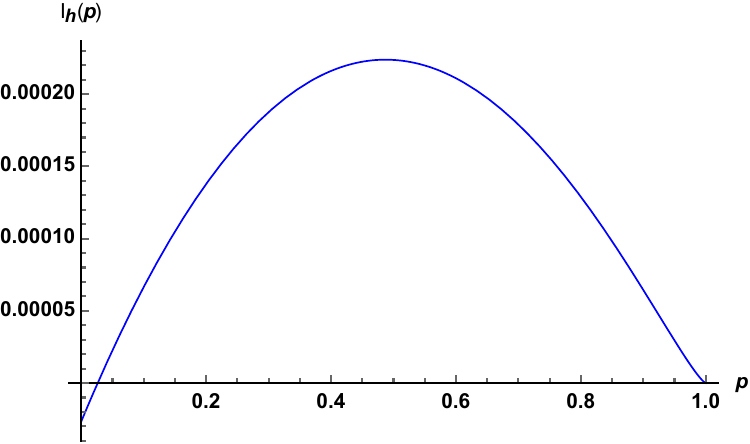}
			\caption{Plot of the integral expression defined in \eqref{dmrlIntEqu_h}.} \label{notconv03}
		\end{center}
	\end{figure}

\end{example}

 The following result shows that the \textit{dmrl} order is preserved by  antistarshaped strictly increasing distortion functions.
\begin{theorem} \label{THdmrl}
Let $X$ and $Y$ be two absolutely continuous random variables with
			strictly increasing distribution functions $F$ and $G,$ respectively. Let $h$ be an antistarshaped strictly increasing distortion function.
	If $X\leq _{dmrl}Y,$ then $X_{h}\leq _{dmrl}Y_{h}$. 
\end{theorem}

\begin{proof}
	Let $\bar{F}$ and $\bar{G}$ be the survival functions of $X$ and $Y,
	$ respectively and let $f$ and $g$ be the respective density functions. Let $%
	X_{h}$ and $Y_{h}$ be associated distorted random variables with respective
	survival funcions $\bar{F}_{h}$ and $\bar{G}_{h}$ and density
	functions $f_{h}$ and $g_{h},$ respectively, given by%
	\begin{eqnarray*}
		\bar{F}_{h}\left( t\right)  &=&h\left( \bar{F}(t)\right) ,\ 
		\bar{G}_{h}\left( t\right) =h\left( \bar{G}(t)\right)  \\
		f_{h}\left( t\right)  &=&h^{\prime }\left( \bar{F}(t)\right) f\left(
		t\right) ,g_{h}\left( t\right) =h^{\prime }\left( \bar{G}(t)\right)
		g(t).
	\end{eqnarray*}%
We first give the idea of the proof. The condition $X\leq _{dmrl}Y$ is equivalent to (\ref{equivdmrl}).  Since, from the assumption on $h$, the function $h\left( 1-t\right) /\left( 1-t\right) $ is increasing on $\left(
0,1\right),$  we would like to apply Lemma \ref{lBP} to  (\ref{equivdmrl}). However, Lemma \ref{lBP} cannot be applied because the integrands  on  (\ref{equivdmrl}) depend on the limits of integration. Instead, we focus on (\ref{ineqdmrl}), which is also equivalent to $X\leq _{dmrl}Y$. The next step is to rewrite (\ref{ineqdmrl}) to satisfy the assumptions of Lemma \ref{lBP}. Note that (\ref{ineqdmrl}) is the same as%
	\begin{equation}
	\int_{q}^{1}\left( 1-t\right) d\left[ \frac{F^{-1}(t)}{g\left( G^{-1}\left(
		p\right) \right) }-\frac{G^{-1}(t)}{f\left( F^{-1}\left( p\right) \right) }%
	\right] \geq 0,\ \ \text{for all }0<p\leq q<1.  \label{fp}
	\end{equation}%
	Now consider the function 
	\begin{equation*}
	\phi \left( s\right) =\int_{s}^{1}I_{\left[ q,1\right] }\left( t\right)
	\left( 1-t\right) d\left[ \frac{F^{-1}(t)}{g\left( G^{-1}\left( p\right)
		\right) }-\frac{G^{-1}(t)}{f\left( F^{-1}\left( p\right) \right) }\right] ,\
	\ 0<s<1,
	\end{equation*}%
	where $I_{\left[ q,1\right] }\left( t\right)=\left\{ 
	\begin{array}{cc}
	0, & 0<t<q \\ 
	   &       \\  
	1, & q\leq t<1,%
	\end{array}%
	\right. 
$ for all $t\in[0,1].$
	Then, by noting that%
	\begin{equation*}
	\phi \left( s\right) =\left\{ 
	\begin{array}{cc}
	\int_{q}^{1}\left( 1-t\right) d\left[ \frac{F^{-1}(t)}{g\left( G^{-1}\left(
		p\right) \right) }-\frac{G^{-1}(t)}{f\left( F^{-1}\left( p\right) \right) }%
	\right] , & 0<s<q \\ 
	&       \\  
	\int_{s}^{1}\left( 1-t\right) d\left[ \frac{F^{-1}(t)}{g\left( G^{-1}\left(
		p\right) \right) }-\frac{G^{-1}(t)}{f\left( F^{-1}\left( p\right) \right) }%
	\right] , & q<s<1,%
	\end{array}%
	\right. 
	\end{equation*}%
	we see that (\ref{fp}) is equivalent to%
	\begin{eqnarray*}
	\int_{s}^{1}I_{\left[ q,1\right] }\left( t\right) \left( 1-t\right) d\left[ 
	\frac{F^{-1}(t)}{g\left( G^{-1}\left( p\right) \right) }-\frac{G^{-1}(t)}{%
		f\left( F^{-1}\left( p\right) \right) }\right]  &\geq &0,\\
	\text{for all } 0<s<1, & &0<p\leq q<1  \notag
	\end{eqnarray*}%
Now we are in conditions to apply  Lemma \ref{lBP}.	Since  $h\left( 1-t\right) /\left( 1-t\right) $ is increasing on $\left(
	0,1\right) ,$ it follows  that%
	\begin{equation*}
	\int_{s}^{1}I_{\left[ q,1\right] }\left( t\right) h(1-t)d\left[ \frac{%
		F^{-1}(t)}{g\left( G^{-1}\left( p\right) \right) }-\frac{G^{-1}(t)}{f\left(
		F^{-1}\left( p\right) \right) }\right] \geq 0,\ 
	\end{equation*}%
	\begin{equation*}
	\ \text{for all }s>0,\ 0<q\leq p<1,
	\end{equation*}%
	or, equivalently,%
	\begin{equation*}
	\int_{q}^{1}h(1-t)d\left[ \frac{F^{-1}(t)}{g\left( G^{-1}\left( p\right)
		\right) }-\frac{G^{-1}(t)}{f\left( F^{-1}\left( p\right) \right) }\right]
	\geq 0,\ \ \text{for all }0<p\leq q<1.
	\end{equation*}%
	This implies%
	\begin{equation*}
	\frac{f\left( F^{-1}\left( p\right) \right) }{g\left( G^{-1}\left( p\right)
		\right) }\geq \frac{\int_{p}^{1}h\left( 1-t\right) dG^{-1}(t)}{%
		\int_{p}^{1}h\left( 1-t\right) dF^{-1}(t)},\ \text{for all }p\in \left(
	0,1\right) ,
	\end{equation*}%
	or, equivalently,%
	\begin{equation*}
	\frac{f\left( F^{-1}\left( 1-h^{-1}\left( 1-p\right) \right) \right) }{%
		g\left( G^{-1}\left( 1-h^{-1}\left( 1-p\right) \right) \right) }\geq \frac{%
		\int_{1-h^{-1}\left( 1-p\right) }^{1}h\left( 1-t\right) dG^{-1}(t)}{%
		\int_{1-h^{-1}\left( 1-p\right) }^{1}h\left( 1-t\right) dF^{-1}(t)},\ \text{%
		for all }p\in \left( 0,1\right) .
	\end{equation*}%
	The change of variable $x=1-h(1-t)$ yields to%
	\begin{eqnarray}
	&&\frac{f\left( F^{-1}(1-h^{-1}\left( 1-p\right) )\right) }{g\left(
		G^{-1}\left( 1-h^{-1}\left( 1-p\right) \right) \right) }  \label{rs} \\
	&\geq &\frac{\int_{p}^{1}\left( 1-x\right) dG^{-1}(1-h^{-1}\left( 1-x\right)
		)}{\int_{p}^{1}\left( 1-x\right) dF^{-1}(1-h^{-1}\left( 1-x\right) )},\ 
	\text{for all }p\in \left( 0,1\right) .  \notag
	\end{eqnarray}%
	Observe that the distribution functions of $X_{h}$ and $Y_{h}$ are given by $%
	F_{h}(x)=1-h\left( 1-F(x)\right) $ and $G_{h}(x)=1-h\left( 1-G(x)\right) .$
	Therefore, (\ref{rs}) is the same as 
	\begin{equation*}
	\frac{f_{h}\left( F_{h}^{-1}\left( p\right) \right) }{g_{h}\left(
		G_{h}^{-1}\left( p\right) \right) }\geq \frac{\int_{p}^{1}\left( 1-x\right)
		dG_{h}^{-1}(x)}{\int_{p}^{1}\left( 1-x\right) dF_{h}^{-1}(x)},\ \text{for
		all }p\in \left( 0,1\right) 
	\end{equation*}%
	which, using Lemma \ref{1}, is equivalent to $X_{h}\leq _{dmrl}Y_{h}.$
\end{proof}

\begin{remark}
	\label{conc}It is well-known that a concave distortion function is antistarshaped. Then, under the assumptions of Theorem \ref{THdmrl}, it follows that $X\leq _{dmrl}Y$ implies $X_{h}\leq _{dmrl}Y_{h}$ for any concave strictly increasing $h$. 
\end{remark}

The  \textit{qmit} order is related to the convex transform order and to the star order (see Arriaza et al\cite{ASS17} for details). Recall that given two non-negative random variables $X$ and $Y$, with distribution functions $F$ and $G$, respectively, $X$ is said to be smaller than $Y$ in the star order (denoted by $X\leq_{*}Y$) if $G^{-1}F(x)$ is starshaped in $x$. Using again that $G^{-1}F(x)=G_h^{-1}F_h(x)$ for any distortion function $h$, we see that the star order (like the convex transform order) is preserved by general distortion functions. Since  
\begin{equation}
\label{imply3}
X\leq_{c}Y \Longrightarrow X\leq_{qmit} Y \Longrightarrow X\leq_{*} Y
\end{equation}
(see Arriaza et al\cite{ASS17}) it is reasonable to wonder whether the order \textit{qmit} is also preserved by general distortion functions. As we show in the following counterexample, again the answer is negative.   
\begin{figure}
	\begin{center}
		\includegraphics*[scale=0.75]{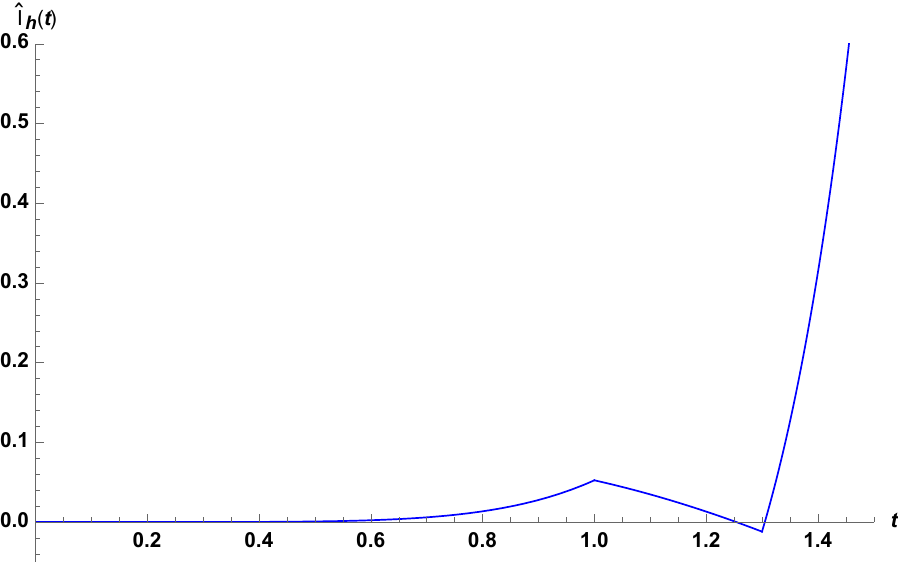}
		\caption{Plot of the function $\hat{I}_h(t)$ defined in \eqref{ex0101}} \label{FigcounterEx02}
	\end{center}
\end{figure}
\begin{example}
\label{ce01}
To construct the counterexample, we need two non-negative random variables such that $X\leq_{qmit}Y$ but $X\nleq_{c}Y$ (otherwise, it follows from the preservation property of the convex transform order and (\ref{imply3}) that $X_h\leq_{qmit}Y_h$ for any distortion $h$). The two random variables considered in Example 1 of Arriaza et al\cite{ASS17} satisfy this condition. Specifically, let $X$ be a non-negative random variable with distribution function $F(x)=1-e^{-\psi(x)}$, where $\psi\in C^{1}(\IR^{+})$ is defined as follows
\begin{equation*}
\psi(x)= \left\{ 
\begin{array}{lll}
e^{x}-1 &  & 0\leq x\leq 1, \\
 &&\\
 2e\sqrt{x}-e-1 &  & 1<x\leq \frac{13}{10},\\
 &&\\
\frac{5}{13}\sqrt{\frac{10}{13}}e^{x^2-\frac{69}{100}}+2(\sqrt{\frac{13}{10}}-1)e-\frac{5}{13}\sqrt{\frac{10}{13}}\,e+e-1 &  &\frac{13}{10}<x,%
\end{array}
\right. 
\end{equation*}
and let $Y$ be a random variable exponentially distributed with rate parameter $\lambda=1$. Consider the distortion function $h(p)=1-(1-p)^5$ and denote
\begin{equation}
\label{eq003}
\hat{I}_h(t)=\displaystyle \int_{0}^{t}[\alpha_{h}^{\prime}(t)-\alpha_{h}^{\prime}(x)]F_{h}(x)\,dt,\,\,\,t>0,
\end{equation}
 \noindent where $\alpha_h(x)=G^{-1}_{h}F_{h}(x)$ and $F_{h}(x)=h^{*}(F(x))$ with $h^{*}(p)=p^5$ the dual distortion of $h(p)$. Observing that $\alpha_h(x)=G^{-1}F(x)=\psi(x)$,  \eqref{eq003} can be rewritten as 
\begin{equation}
\label{ex0101}
\hat{I}_h(t)=\displaystyle \int_{0}^{t}[\psi^{\prime}(t)-\psi^{\prime}(x)]F^{5}(x)\,dt,\,\,\,t>0.
\end{equation}
We know, from Theorem 4 in Arriaza et al\cite{ASS17}, that $X_{h}\leq_{qmit}Y_{h}$ if and only if $\hat{I}_h(t)\ge 0$ for all $t$. However, this condition fails in the interval $[1.2539,1.3050]$ (see Figure \ref{FigcounterEx02}) and therefore $X_{h}\nleq_{qmit} Y_{h}$.

\end{example}
The following result can be proved following the same lines as the proof of Theorem \ref{THdmrl}.

\begin{theorem}
	\label{qmit}
	Let $X$ and $Y$ be two absolutely continuous random variables. Let $h$ be a strictly increasing distortion functions such that its dual $h^*$ is antistarshaped. If $X\leq
	_{qmit}Y,$ then $X_{h}\leq _{qmit}Y_{h}.$ 
\end{theorem}

\begin{remark}
\label{qmitcx}
It is easy to see that the dual distortion of a convex and strictly increasing distortion is a concave function, and therefore, antistarshaped. Then, under the assumptions of Theorem \ref{qmit}, it follows that $X\leq _{qmit}Y$ implies $X_{h}\leq _{qmit}Y_{h}$ for any convex strictly increasing $h$.
\end{remark}


\section{Stochastic comparisons in coherent systems with exchangeable components}
\label{section3}

In this section, we apply the results in Section \ref{main} to study the preservation of the stochastic orders defined in Definition \ref{fourorders} under the formation of coherent systems with ID components. First, we note that the results can be trivially applied in the case of parallel and series systems. For example, let $(X_{1},...,X_{n})$ be a random vector with ID components  and distributional copula  $\hat{C}.$ It is easy to see that the reliability function of the parallel system $T_1=\max\{X_{1},...,X_{n}\}$ satisfies $\bar F_{T_1}=h(\bar F(t)),$ where $\bar F$ is the common reliability function of the components and \begin{equation}\label{hnn}
h(p)=1-\hat{C}(1-p,\ldots,1-p)
\end{equation}
is a distortion function. Similarly, let $(Y_{1},...,Y_{n})$ be a random vector with ID components and the same distributional copula  $\hat{C}$. The reliability function  of the system $T_2 =\max\{Y_{1},...,Y_{n}\}$ satisfies $\bar G_{T_2}=h(\bar G(t)),$ where  $\bar G$ is the common reliability function of the components. If  (\ref{hnn}) is antistarhaped, it follows from theorems \ref{THew} and \ref{THdmrl} that  $X_1\leq _{ew[dmrl]}Y_1$ implies $T_1 \leq _{ew[dmrl]} T_2$. Examples of copulas such that  (\ref{hnn}) is concave (and, therefore, antistarshaped) are the independence copula ($		\hat{C}\left( p_{1},...,p_{n}\right) =p_{1}\cdot \cdot\cdot p_{n}$), the comonotone copula ($\hat{C}\left( p_{1},...,p_{n}\right) =\min
\left\{ p_{1},...,p_{n}\right\}$) and, for $n=2$, the bivariate Cuadras-Aug\'{e} family of copulas $(\hat{C}_{\theta }(p_{1},p_{2})=\left[ \min \left( p_{1},p_{2}\right) %
\right] ^{\theta }\left[ p_{1}p_{2}\right] ^{1-\theta }, \ \ 0<\theta <1)$.
A similar argument can be used in the case of series systems taking into account that the reliability function of the system $T =\min\{X_{1},...,X_{n}\}$ satisfies $\bar F_{T}=g(\bar F(t)),$ where  
$g(p)=C(p,\ldots,p)$ is a distortion function and $C$ is the corresponding survival copula.	
	
Our purpose now is to provide applications in the case of more complex systems by using non-convex starshaped (respectively, non-concave antistarshaped) distortion functions. To that end, we consider coherent system with exchangeable components.  Navarro et al\cite{NavarroSandoval2007} proved that the reliability function of any coherent system with exchangeable components can be expressed as a mixture of series systems. Let $T=\phi(X_1,\ldots,X_n)$ be the lifetime of a system with $n$ exchangeable components, then its reliability function can be written as:  
\begin{equation} 
\label{minimalSig}
\bar{F}_{T}(t) = \displaystyle \sum_{i=1}^{n}a_i\bar{F}_{1:i}(t)
\end{equation}
\noindent where the vector $(a_1,a_2,\ldots,a_n)$ is called the minimal signature associated to the system, and $\bar{F}_{1:i}$ is the reliability function of the series systems formed by $i$ components. Since the vector $(X_1,X_2,\ldots,X_n)$ is exchangeable we can express $\bar{F}_{1:i}$ as follows:
\begin{equation*}
\bar{F}_{1:i}(t) =  \Pr[X_{\pi(1)}>t,X_{\pi(2)}>t,\ldots,X_{\pi(i)}>t]
\end{equation*}  
\noindent for any permutation $\pi:\{1,2,\ldots,n\}\rightarrow\{1,2,\ldots,n\}$. In particular, we can write:
\begin{equation}
\label{F1iExch}
\bar{F}_{1:i}(t)=\Pr[X_{1}>t,X_{2}>t,\ldots,X_{i}>t] =  C(\bar{F}(t),\overset{(i)}{\cdots},\bar{F}(t),1,\overset{(n-i)}{\cdots},1)
\end{equation}
\noindent where $\bar{F}(t)$ is the common reliability function of $X_i$ for all $i=1,2,\ldots,n$ and $C$ is the symmetric survival copula associated. Replacing 
\eqref{F1iExch} in \eqref{minimalSig} we get:
\begin{equation}
\label{RelFunDistExch}
\bar{F}_{T}(t) = h_{T}(\bar{F}(t))
\end{equation}
\noindent where $h_{T}(p)=\displaystyle \sum_{i=1}^{n}a_i h_{1:i}(p)$ and $h_{1:i}(p) = C(p, \overset{(i)}{\cdots}  ,p,1, \overset{(n-i)}{\cdots}   ,1)$. We note that $h_T$ only depends on the minimal signature associated to the system and its corresponding symmetric survival copula.

\subsection{Multivariate Durante et al. copula}

We consider the multivariate symmetric copula defined by Durante et al.\cite{Duranteetal2007}

Given a function $f:[0,1]\rightarrow [0,1]$, let $C_{f}:[0,1]\times\overset{(n)}{\cdots}\times[0,1]\rightarrow [0,1]$ be the function defined by
\begin{equation}
\label{copulaDurante}
C_{f}(p_1,p_2,\ldots,p_n) = p_{[1]}\,\displaystyle \prod_{i=2}^{n}f(p_{[i]})
\end{equation}
\noindent where $p_{[1]}, p_{[2]},\ldots, p_{[n]}$ denote the components of the vector $(p_1,\ldots,p_n)$ rearranged in increasing order. Thus, $p_{[1]}$ and $p_{[n]}$ represent the minimum and maximum of $(p_1,\ldots,p_n)$, respectively. Durante et al.\cite{Duranteetal2007} proved that $C_f(p_1,\ldots,p_n)$ is an $n$-dimensional copula if, and only if, the function $f$ satisfies the following properties:
\begin{align*}
i)\,\,& f(1)=1;\\
ii)\,\,& f\,\,\mbox{is increasing in}\,\, [0,1];\\
iii)\,\,& f\,\,\mbox{is antistarshaped in}\,\, (0,1].
\end{align*}
From different choices of $f$ we can obtain some particular copulas well-known in the literature. For example, if $f(p)=p$, then $C_{f}$ coincides with the independence copula. If $f_{\gamma}(p)=\gamma\,p+(1-\gamma)$ with $\gamma\in (0,1)$ and $n=2$, then we obtain the Fr\'{e}chet family copulas with parameter $\gamma$. A possible multivariate version of the Cuadras-Aug\'e can be obtained by taking $f_{\gamma}(p)=p^\gamma$ with $\gamma\in [0,1].$

The following result provides conditions  under which the distortion function associated to the reliability function of a coherent system with exchangeable components and a survival copula given by \eqref{copulaDurante} is  starshaped (antistarshaped).

\begin{theorem}
\label{ThCopulaDurante}
Let $T=\phi(X_1,X_2,\ldots,X_n)$ be the lifetime of a coherent system with $n$ exchangeable components, minimal signature $(a_1,a_2,\ldots,a_n)$ and survival copula $C_{f}$ given by \eqref{copulaDurante} with $f\in C^{1}(0,1)$. Then, 
\begin{equation}
\label{formulaCopulaDurante}
h_{T}(p)\,\,\mbox{is starshaped [antistarshaped]}\,\,\Longleftrightarrow \displaystyle \sum_{k=1}^{n-1}k\,a_{k+1}\,f^{k-1}(p)\geq [\leq]\, 0\,\,\mbox{for all}\,\, p\in (0,1).
\end{equation}
\noindent 
\end{theorem}
\begin{proof}
Since the system has exchangeable components, its reliability function takes the form \eqref{RelFunDistExch}, where
\begin{equation}
\label{htTH}
h_T(p)  = \displaystyle \sum_{k=1}^{n}a_k h_{1:k}(p)=  \displaystyle \sum_{k=1}^{n}a_k\, C_{f}(p, \overset{(k)}{\cdots}  ,p,1, \overset{(n-k)}{\cdots}   ,1)=
\displaystyle \sum_{k=1}^{n}a_k\,p\,f^{k-1}(p)
\end{equation}
\noindent for all $p\in [0,1].$ From \eqref{htTH} we conclude that $h_{T}(p)/p$ is increasing [decreasing] in $p$ if, and only if, \eqref{formulaCopulaDurante} holds. 
\end{proof}

In particular, for systems with 3 and 4 components, we have the following results (we omit the straightforward proofs).
\begin{corollary}
\label{stcoro}
Let $T=\phi(X_1,X_2,X_3)$ be the lifetime of a coherent system with $3$ exchangeable components, minimal signature $(a_1,a_2,a_3)$ and survival copula $C_{f}$ given by \eqref{copulaDurante} with $f\in C^{1}(0,1)$.  Define $\omega =\frac{-a_2}{2a_3}$. Then,
\begin{equation*}
\text{if\,} 
\begin{cases}
a_3>0 \begin{cases} 
	\omega \geq 1 \hspace{3.12cm}\longrightarrow \, h_T(p)\,\, \mbox{is antistarshaped for any function}\,\, f,\\
 	\omega \in (0,1),\ f(0)\geq \omega  \hspace{0,5cm}\longrightarrow\,h_T(p)\,\, \mbox{is starshaped},\\
  	\omega \leq 0 \hspace{3.12cm}\longrightarrow\, h_T(p)\,\, \mbox{is starshaped for any function}\,\, f, \end{cases}\\
a_3<0 \begin{cases} 
	\omega \geq 1 \hspace{3.12cm}\longrightarrow\, h_T(p)\,\, \mbox{is starshaped for any function}\,\, f,\\ 
	\omega \in (0,1), \ f(0)\geq \omega  \hspace{0.5cm}\longrightarrow\, h_T(p)\,\, \mbox{is antistarshaped},\\
	\omega \leq 0 \hspace{3.12cm}\longrightarrow \, h_T(p)\,\, \mbox{is antistarshaped for any function}\,\, f.\end{cases}\\
\end{cases}
\end{equation*}
\end{corollary}

\begin{corollary}
\label{sndcoro}
Let $T=\phi(X_1,X_2,X_3,X_4)$ be the lifetime of a coherent system with $4$ exchangeable components, minimal signature $(a_1,a_2,a_3,a_4)$ and survival copula  $C_{f}$ given by \eqref{copulaDurante} with $f\in C^{1}(0,1)$. Define $\Delta =a_3^2-3a_2a_4$, $x_{[1]}=\min\{x_1,x_2\}$ and $x_{[2]}=\max\{x_1,x_2\}$, where $x_1$ and $x_2$ are the real roots of the equation $3a_4x^2+2a_3x+a_2=0$. Then,
\begin{equation*}
\text{if\,} 
\begin{cases}
a_4>0 \begin{cases} 
	\Delta \leq 0 \hspace{5.25cm}\longrightarrow \, h_T(p)\,\, \mbox{is starshaped for any function}\,\, f,\\
  	\Delta > 0 \begin{cases} 
	x_{[2]}\leq 0 \hspace{3.45cm}\longrightarrow\, h_T(p)\,\, \mbox{is starshaped for any function}\,\, f,\\	
	x_{[2]}\in(0,1), \ f(0)\geq x_{[2]} \hspace{0.55cm}\longrightarrow\, h_T(p)\,\, \mbox{is starshaped},\\	
	x_{[1]}\leq 0,\,x_{[2]}\geq 1 \hspace{1.9cm}\longrightarrow\, h_T(p)\,\, \mbox{is antistarshaped for any function}\,\, f,\\	
	x_{[1]}\in (0,1),\,x_{[2]}\geq 1, f(0)\geq x_{[1]} \hspace{0.1cm}\longrightarrow\, h_T(p)\,\, \mbox{is antistarshaped},\\
	x_{[1]}\geq 1 \hspace{3.5cm}\longrightarrow\, h_T(p)\,\, \mbox{is starshaped for any function}\,\, f,\\	
	\end{cases}	
	\end{cases}\\
a_4<0 \begin{cases} 
	\Delta \leq 0 \hspace{5.25cm}\longrightarrow\, h_T(p)\,\, \mbox{is antistarshaped for any function}\,\, f,\\ 
  	\Delta > 0 \begin{cases} 
	x_{[2]}\leq 0 \hspace{3.45cm}\longrightarrow\, h_T(p)\,\, \mbox{is antistarshaped for any function}\,\, f,\\	
	x_{[2]}\in(0,1), \ f(0)\geq x_{[2]} \hspace{0.6cm}\longrightarrow\, h_T(p)\,\, \mbox{is antistarshaped},\\	
	x_{[1]}\leq 0,\,x_{[2]}\geq 1 \hspace{1.9cm}\longrightarrow\, h_T(p)\,\, \mbox{is starshaped for any function}\,\, f,\\	
	x_{[1]}\in (0,1),\,x_{[2]}\geq 1,  f(0)\geq x_{[1]} \hspace{0.1cm}\longrightarrow\,h_T(p)\,\, \mbox{is starshaped},\\
	x_{[1]}\geq 1 \hspace{3.5cm}\longrightarrow\, h_T(p)\,\, \mbox{is antistarshaped for any function}\,\, f,\\	
	\end{cases}	
	\end{cases}
\end{cases}
\end{equation*}
\noindent if  $a_4=0$, see the scheme provided in Corollary \ref{stcoro}.
\end{corollary}

\begin{example}
Let $T_1=\max(X_1,X_2,\min(X_3,X_4))$ and $T_2=\max(Y_1,Y_2,\min(Y_3,Y_4))$ be the lifetimes of two coherent systems with $4$ exchangeable components and survival copula $C_{f}$ given by \eqref{copulaDurante}. The minimal signature (see Navarro et al.\cite{NavarroSandoval2007}) is given by the vector $\mathbf{a}=(2, 0, -2, 1)$. Let us assume that the random vectors $\mathbf{X}=(X_1,X_2,X_3,X_4)$ and $\mathbf{Y}=(Y_1,Y_2,Y_3,Y_4)$ have the same survival copula $C_{f}$ defined in  \eqref{copulaDurante}, then both distortion functions $h_{T_1}$ and $h_{T_2}$ are given by 
\begin{equation}
\label{eq001101}
h_{T_1}(p)=h_{T_2}(p)= 2\,p-2\,p\,f^2(p)+p\,f^3(p).
\end{equation}
Since $a_4>0$, $\Delta = 4$, $x_{[1]}=0$ and $x_{[2]}=4/3$, it follows from Corollary \ref{sndcoro} that \eqref{eq001101} is antistarshaped independently of the choice of $f$. It follows from theorems \ref{THew} and \ref{THdmrl} that $ X_1 \leq_{ew[dmrl]} Y_1$ implies $T_{1}\leq_{ew[dmrl]}T_{2}$.
\end{example}

\begin{example}
Let $T_1$ and $T_2$ be the lifetimes of two coherent systems with $4$ exchangeable components, where $T_1=\min(X_1,\max(X_2,X_3),\max(X_2,X_4))$ and $T_2=\min(Y_1,\max(Y_2,Y_3),\max(Y_2,Y_4))$. The minimal signature associated to $T_1$ and $T_2$ is given by the vector $\mathbf{a}=(0,1,1,-1)$. Let us assume that the random vectors $\mathbf{X}=(X_1,X_2,X_3,X_4)$ and $\mathbf{Y}=(Y_1,Y_2,Y_3,Y_4)$ have the same survival copula $C_{f}$ defined in  \eqref{copulaDurante}, then both distortion functions $h_{T_1}$ and $h_{T_2}$ are given by 
\begin{equation}
\label{eq001100}
h_{T_1}(p)=h_{T_2}(p)= p\,f(p)+p\,f^2(p)-p\,f^3(p).
\end{equation}
Since $a_4<0$, $\Delta = 4$, $x_{[1]}=-1/3$ and $x_{[2]}=1$, it follows from Corollary \ref{sndcoro} that \eqref{eq001100} is starshaped independently of the choice of $f$. It follows from Theorem \ref{THttt} that $X_1\leq_{ttt}Y_1$ implies $T_{1}\leq_{ttt}T_{2}$.
\end{example}
\subsection{Multivariate copula with given diagonal section}

Given an $n$-dimensional copula $C:[0,1]\times\overset{(n)}{\cdots}\times[0,1]\rightarrow [0,1]$, we define the diagonal section of the copula $C$ as the function $\delta_{C}:[0,1]\rightarrow [0,1]$ such that $\delta_{C}(p)=C(p,\ldots,p)$ for all $p\in [0,1]$. The construction of an $n$-dimensional copula from a diagonal function is a relevant problem in copula theory.

\begin{definition}
A function $\mathfrak{d}:[0,1]\longrightarrow [0,1]$ is called an $n$-dimensional diagonal function, if it satisfies the following properties:
\begin{align*}
a) & \,\,\,\mathfrak{d}(1)=1,\\
b) & \,\,\,\mathfrak{d}(p)\leq p,\,\,\mbox{for all}\,\, p\in[0,1],\\
c) & \,\,\,0\leq \mathfrak{d}(p_2)- \mathfrak{d}(p_1)\leq n\,(p_2-p_1)\,\,\mbox{for all}\,\, p_1,\,p_2\in[0,1]\,\,\mbox{such that}\,\, p_1\leq p_2.
\end{align*}
\end{definition}

The set of all $n$-dimensional diagonals will be denoted by $\wp_n$. It is clear that any diagonal section $\delta_{C}$ of an $n$-dimensional copula $C$ belongs to the set $\wp_n$. The reverse result was proved by Cuculescu and Theodorescu\cite{Cuculescu2001}, i.e., given any $n$-dimensional diagonal $\mathfrak{d}\in \wp_n$, there exists an $n$-dimensional copula $C$ such that the diagonal section $\delta_{C}(p)=\mathfrak{d}(p)$ for all $p\in [0,1]$. Jaworski\cite{Jaworski2009} provides a constructive method to find the corresponding $n$-dimensional copula given an $n$-dimensional diagonal function.  
\begin{theorem}[Jaworski\cite{Jaworski2009}] If $\mathfrak{d}\in \wp_n$, then there exists an $n$-dimensional copula $C_{\mathfrak{d}}$ such that $\delta_{C_{\mathfrak{d}}}(p,\ldots,p)=\mathfrak{d}(p)$ for all $p\in [0,1]$, with
\begin{equation}
\label{JaworskiCopula}
C_{\mathfrak{d}}(p_1,p_2,\ldots,p_n)=\frac{1}{n}\,\displaystyle\sum_{i=1}^{n}\min\{f(p_{\tau^{i}(1)}),f(p_{\tau^{i}(2)}),\ldots,f(p_{\tau^{i}(n-1)}),\mathfrak{d}(p_{\tau^{i}(n)})\}
\end{equation}
\noindent where $f:[0,1]\longrightarrow [0,1]$ is the function given by
\begin{equation}
\label{funcionMCP}
f(u) = \frac{nu-\mathfrak{d}(u)}{n-1}
\end{equation}
\noindent and $\tau^{i}:\{1,2,\ldots,n\}\longrightarrow\{1,2,\ldots,n\}$ is a permutation defined as $\tau^{i}(k)=k+i$ mod $n$, for all $i=1,2,\ldots,n$.
\end{theorem}

If we consider $n=2$ in \eqref{JaworskiCopula}, we obtain the bivariate copula provided by Fredricks and Nelsen\cite{FredricksNelsen97} and defined as
\begin{equation*}
C_{\mathfrak{d}}(p_1,p_2)=\min\{p_1,p_2,\frac{\mathfrak{d}(p_1)+\mathfrak{d}(p_2)}{2}\}.
\end{equation*}
Observe that any $n$-dimensional copula defined as \eqref{JaworskiCopula} is symmetric. Therefore, the distortion function associated to any coherent system with $n$ ID components and survival copula $C_{\mathfrak{d}}$ can be expressed as
\begin{equation}
\label{distCopDiag}
h_{T}(p) = \displaystyle \sum_{i=1}^{n}a_i\,C_{\mathfrak{d}}(p,\overset{(i)}{\cdots},p,1,\overset{(n-i)}{\cdots},1)
\end{equation}
\noindent where $\mathbf{a}=(a_1,a_2,\ldots,a_n)$ is the corresponding minimal signature of the system. Since $\mathfrak{d}(u)\leq f(u)$ for all $u\in [0,1]$ and $\mathfrak{d}\in \wp_{n}$, we deduce that 
\begin{equation}
\label{copulaSimplif}
C_{\mathfrak{d}}(p,\overset{(i)}{\cdots},p,1,\overset{(n-i)}{\cdots},1) = \frac{1}{n}\left [(n-i)f(p)+i\,\mathfrak{d}(p)\right ].
\end{equation}
Replacing \eqref{funcionMCP} and \eqref{copulaSimplif} into \eqref{distCopDiag}, we obtain that
\begin{equation}
\label{repSySDiagCopula}
h_T(p)= \alpha\,p+\beta\,\mathfrak{d}(p)
\end{equation}
\noindent where $\alpha = \frac{1}{n-1}\displaystyle \sum_{i=1}^n a_i(n-i)$ and $\beta = \frac{1}{n-1}\displaystyle \sum_{i=1}^n a_i(i-1)$. Note that $\alpha$ and $\beta$ only depend on the minimal signature of the system and on the number of components. From the representation \eqref{repSySDiagCopula} we deduce the following theorem. 
\begin{theorem}
\label{ThDiagCop}
Let $T=\phi(X_1,X_2,\ldots,X_n)$ be the lifetime of a coherent system with $n$ exchangeable components and survival copula $C_{\mathfrak{d}}$ given by \eqref{JaworskiCopula} with $\mathfrak{d}\in \wp_n$ an $n$-dimensional diagonal, then the distortion function associated to the system $h_{T}(p)$ can be expressed as \eqref{repSySDiagCopula} and it satisfies that
$$h_T(p)\,\, \mbox{is starshaped [antistarshaped]}\,\,\Longleftrightarrow\,\, \mathfrak{d}(p)\,\, \mbox{is starshaped and $\beta>0\,\,[\beta<0]$}.$$
\end{theorem}
\begin{example}
Let $T_1=\min(X_1,X_2,\max(\min(X_3,X_4),\min(X_3,X_5),\min(X_4,X_5)))$ and $T_2=\min(Y_1,Y_2,\max(\min(Y_3,Y_4),\min(Y_3,Y_5),\min(Y_4,Y_5)))$ be the lifetimes of two coherent systems with $5$ exchangeable components. The minimal signature associated to both systems is given by the vector $\mathbf{a}=(0,0,0,3,-2)$. Let us assume that the random vectors $\mathbf{X}=(X_1,X_2,X_3,X_4,X_5)$ and $\mathbf{Y}=(Y_1,Y_2,Y_3,Y_4,Y_5)$ have the same survival copula $C_{\mathfrak{d}}$ defined as \eqref{JaworskiCopula}, with diagonal given by the starshaped function $\mathfrak{d}(p)=2p^2-p^3$. From \eqref{repSySDiagCopula} the distortion function associated to both systems is
\begin{equation*}
h_{T_1}(p)=h_{T_2}(p)=\frac{3}{4}p+\frac{1}{4}\mathfrak{d}(p).
\end{equation*}
From Theorem \ref{ThDiagCop} we conclude that $h_{T_1}(p)$ is a starshaped function (which is non convex). It follows from Theorem \ref{THttt} that $X_1\leq_{ttt}Y_1$ implies $T_{1}\leq_{ttt}T_{2}$.
\end{example}
\begin{example}
A k-out-of-n system is a system with n components which fails if, and only if, at least k components fail.
Let $X_{3:4}$ and $Y_{3:4}$ be the lifetimes of two $3$-out-of-$4$ systems with $4$ exchangeable components. The minimal signature associated to both systems is given by the vector $\mathbf{a}=(0,6,-8,3)$. Let us assume that the random vectors $\mathbf{X}=(X_1,X_2,X_3,X_4)$ and $\mathbf{Y}=(Y_1,Y_2,Y_3,Y_4)$ have the same survival copula $C_{\mathfrak{d}}$ given by \eqref{JaworskiCopula}, with diagonal function given by $\mathfrak{d}(p)=\frac{1}{4}p+\frac{3}{4}(2p^2-p^3)$. Then, from \eqref{repSySDiagCopula} the distortion functions associated to both systems are given by
\begin{equation*}
h_{3:4}(p)=\frac{4}{3}p-\frac{1}{3}\mathfrak{d}(p).
\end{equation*}
The diagonal function $\mathfrak{d}(p)$ is a starshaped (non convex) function. It follows from Theorem \ref{ThDiagCop} that $h_{3:4}(p)$ is antistarshaped and from theorems \ref{THew} and \ref{THdmrl} that $X_1\leq_{ew[dmrl]}Y_1$ implies $X_{3:4}\leq_{ew[dmrl]}Y_{3:4}$.
\end{example}

\begin{example}
\label{exampleqmit002}
Let $T_1=\min(X_1,X_2,\max(X_3,X_4))$ and $T_2=\min(Y_1,Y_2,\max(Y_3,Y_4))$ be the lifetimes of two coherent systems with $4$ exchangeable components. The minimal signature associated to both systems is given by the vector $\mathbf{a}=(0,0,2,-1)$. Let us assume that the random vectors $\mathbf{X}=(X_1,X_2,X_3,X_4)$ and $\mathbf{Y}=(Y_1,Y_2,Y_3,Y_4)$ have the same survival copula $C_{\mathfrak{d}}$ given by \eqref{JaworskiCopula}, with diagonal function given by $\mathfrak{d}(p)=1-\frac{7}{4}(1-p)+\frac{3}{2}(1-p)^2-\frac{3}{4}(1-p)^3$ (see right plot of Figure \ref{Exampleqmit002}). From \eqref{repSySDiagCopula} we obtain that the distortion functions associated to both systems are given by
\begin{equation*}
h_{T_1}(p)=h_{T_2}(p)=\frac{2}{3}p+\frac{1}{3}\mathfrak{d}(p).
\end{equation*}
It is not difficult to see that $h(p)$ is a strictly increasing distortion and that its dual distortion $h^{*}(p)=1-\frac{2}{3}(1-p)-\frac{1}{3}\mathfrak{d}(1-p)$ is an antistarshaped (non concave) function. Figure \ref{Exampleqmit002} plots the functions $h^{*}(p)/p$ (left) and  $\mathfrak{d}(p)$ (right). From Theorem \ref{qmit} it follows that $X_1\leq_{qmit}Y_1$ implies $T_{1}\leq_{qmit}T_{2}$.
\end{example}
\begin{figure}
	\begin{center}
		\includegraphics*[scale=0.33]{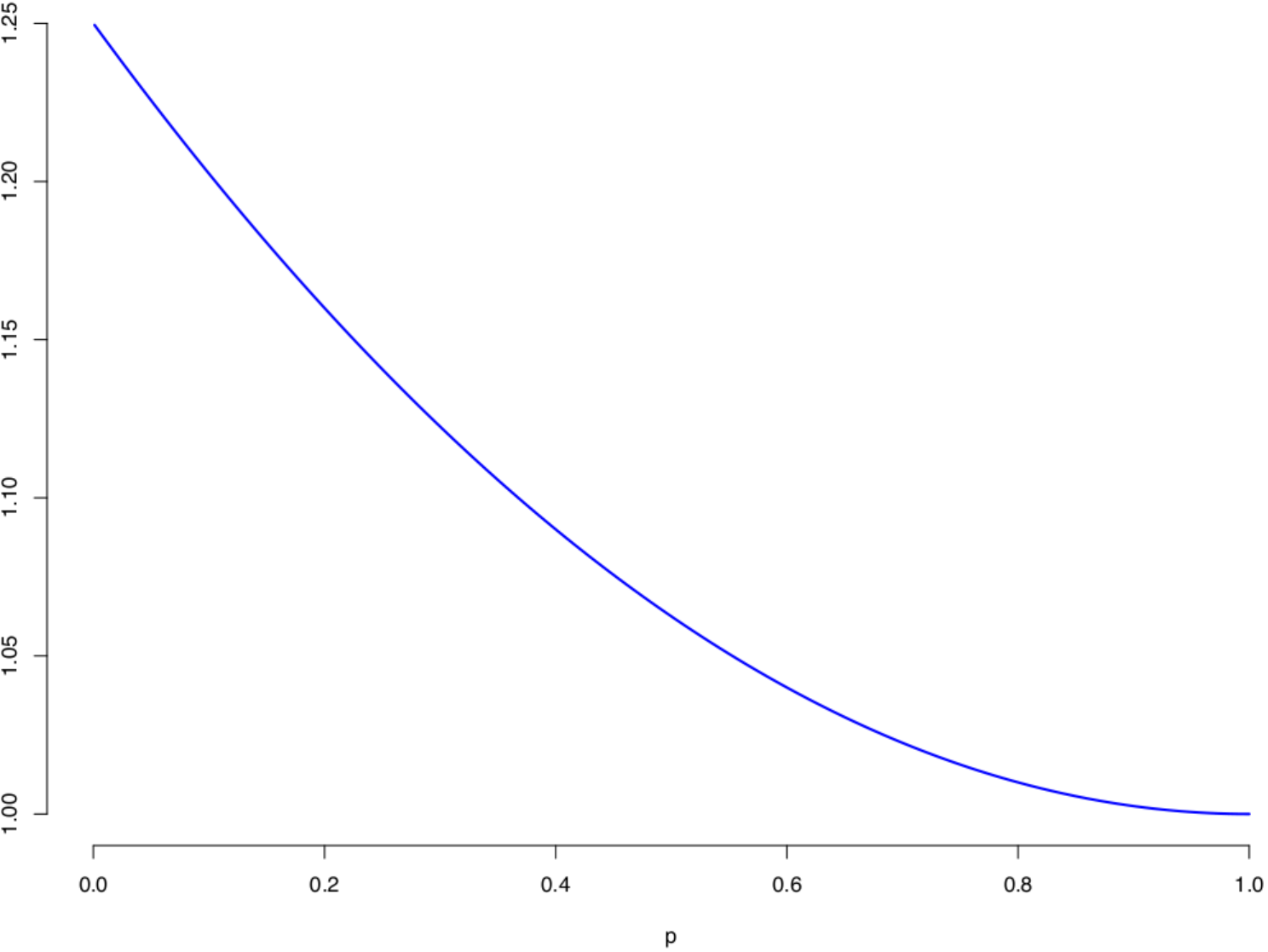}
	    \includegraphics*[scale=0.33]{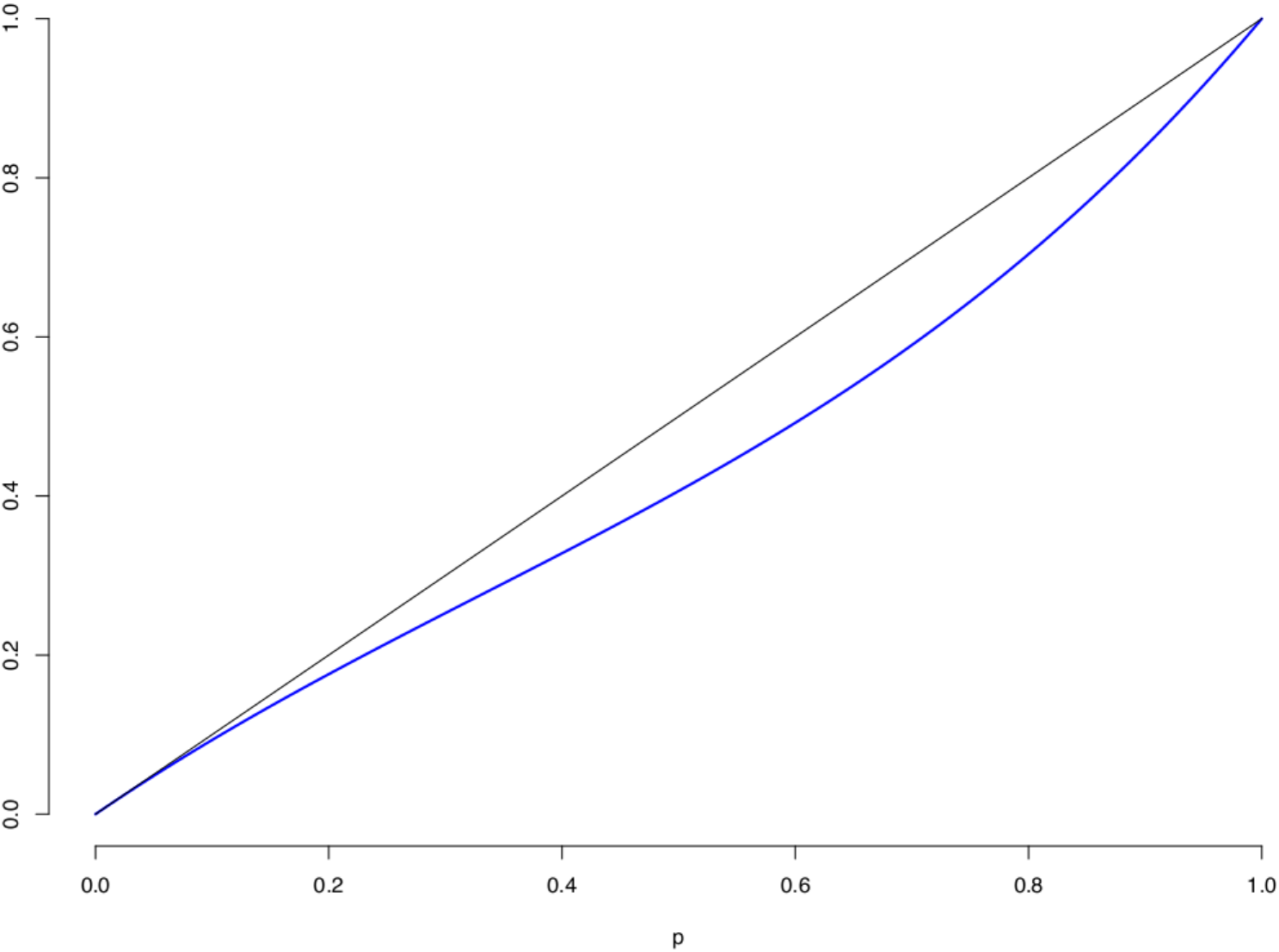}
		\caption{Plots of the functions $h^{*}(p)/p$ (left),  $\mathfrak{d}(p)$ (right in blue) and the identity function (right in black) corresponding to the Example \ref{exampleqmit002}.} \label{Exampleqmit002}
	\end{center}
\end{figure}

\section{Conclusions} 
In this paper, we have shown that the total time on test transform order, the excess wealth order, the decreasing mean residual life order and the quantile mean inactivity time order are preserved by some classes of distortion functions. In particular, the results regarding the two first orders extend previous results in the literature by considering starshaped and antistarshaped distortion functions, which contain, respectively, the classes of convex and concave distortion functions. The results have been applied to study the preservation of these four stochastic orders under the formation of coherent systems. After showing that the results can be directly applied to the study of parallel and series systems with ID components, we have focussed on more complex systems with exchangeable components. We have illustrated the applicability of the results with several examples using different families of copulas.  The relevant distortion functions in the examples are non-convex starshaped and non-concave antistarshaped, which highlight the usefulness of the theoretical results.

\section*{Acknowledgements}

We thank the anonymous reviewers for helpful comments. The authors were supported by Ministerio de Economía y Competitividad of Spain under grant MTM2017-89577-P.\\



\begin{thebibliography}{99}

\bibitem{SS07}
Shaked M, Shanthikumar JG. Stochastic orders. New York: Springer Series in Statistics. Springer; 2007.


\bibitem{KLS2002}	
Kochar SC, Li X, Shaked M. The total time on test transform and the excess wealth stochastic orders of distributions. Adv  Appl Probab 2002;34:826--856.

\bibitem{LC2004} Li X, Chen J. Aging properties of the residual life length of k-out-of-n systems 	with independent but non-identical components. Appl Stoch Model Bus Ind 2004;20:143--153.

\bibitem{BMRS2016}Belzunce  F, Martínez-Riquelme C, Ruiz JM, Sordo MA. On sufficient conditions for the comparison
in the excess wealth order and spacings. J Appl Probab. 2016;53:33--46.

\bibitem{BMRS2017}Belzunce  F, Martínez-Riquelme C, Ruiz JM, Sordo MA. On the Comparison of Relative Spacings with Applications. Methodol Comput Appl Probab 2017;19:357--376.

\bibitem{SP2017} Sordo MA, Psarrakos G. Sochastic comparisons of interfailure times under a relavation replacement policy. J. Appl. Prob. 2017; 54: 134-145.


\bibitem{BP75} %
Barlow RE, Proschan F. %
Statistical theory of reliability and life testing. International
Series in Decision Processes. New York: Holt, Rinehart and Winston, Inc; 1975.


\bibitem{KocharWiens1987} %
Kochar SC, Wiens DP.
Partial orderings of life distributions with respect to their aging properties. 
Nav Res Log 1987;34:823--829.


\bibitem{ASS17}
Arriaza A, Sordo MA, Su\'arez-Llorens A. 
Comparing residual lives and inactivity times by transform stochastic orders. 
IEEE T Reliab 2017;66:366--372.


\bibitem{ABM19}
Arriaza A, Belzunce F, Martínez-Riquelme C. %
Sufficient conditions for some transform orders based on
the quantile density ratio. Methodol Comput Appl Probab 2019; https://doi.org/10.1007/s11009-019-09740-6.


\bibitem{Kayid2018} %
Kayid M, Izadkhah S, Alfifi A. Increasing Mean Inactivity Time Ordering: A Quantile Approach. Math Probl Eng 2018; 10 pages.


\bibitem{NASS16}  %
Navarro J, del Aguila Y,  Sordo MA, Su\'arez-Llorens A. Preservation of stochastic orders under the formation of generalized distorted distributions. Applications to coherent systems. Methodol Comput Appl 2016;18:529--545.

\bibitem{MN17}  Miziula P, Navarro J. Sharp bounds for the reliability of systems and mixtures with ordered components. Nav Res Log 2017;64:108--116.


\bibitem{N18}%
Navarro J. Stochastic comparisons of coherent systems. Metrika 2018;81:465--482.

\bibitem{NavRych}
Navarro J, Rychlik T. Comparisons and bounds for expected lifetimes of reliability systems. Eur J Oper Res 2010;207:309--317.


\bibitem{NASS13}%
Navarro J, del Aguila Y, Sordo MA,  Su\'arez-Llorens A.
Stochastic ordering properties for systems with dependent identically distributed components. Appl Stoch Model Bus Ind 2013;29:264--278.

\bibitem{NASS14}
Navarro J, del Aguila Y, Sordo MA,  Su\'arez-Llorens A.
Preservation of reliability classes under the formation of coherent systems
Appl Stoch Model Bus Ind 2014;30:444--454.

\bibitem
{NG16}%
Navarro J, Gomis MC. %
Comparisons in the mean residual life order of coherent systems with
identically distributed components. Appl Stoch Model Bus Ind 2016;32:33--47.

\bibitem{NA17}%
Navarro J, del Aguila Y. Stochastic comparisons of distorted distributions, coherent systems and mixtures with ordered components. Metrika 2017;80:627--648.


\bibitem{ANS18}
Arriaza A, Navarro J, Su\' arez-Llorens A. %
Stochastic comparisons of replacement policies in coherent
systems under minimal repair. Nav Res Log 2018;65:550--565.


\bibitem{NS2018}
Navarro J, Sordo MA. Stochastic comparisons and bounds for
conditional distributions by using copula properties. Depend Model 2018;6:156--177.

\bibitem{NAS18}
Navarro J, Arriaza A, Su\' arez-Llorens A. Minimal repair of failed components in coherent systems. Eur J Oper Res 2019;279:951--964. https://doi.org/10.1016/j.ejor.2019.06.013.

\bibitem{LiLi2018}
Li C, Li X. Preservation of increasing convex/concave order under the formation of parallel/series system of dependent components. Metrika. 2018; 81:445--464.


\bibitem{NC2019}
Navarro J, Cal\`{i} C. Inactivity times of coherent systems with dependent components under periodical inspections. Appl Stoch Model Bus Ind 2019;35:871--892. https://doi.org/10.1002/asmb.2416.


\bibitem{LS2007}
Li X, Shaked M. A general family of univariate stochastic orders. J Stat Plan Infer 2007;137:3601--3610.

\bibitem{Nelsen06}
Nelsen RB. An introduction to copulas. Second Edition. New York: Springer Series$+$Business Media, Inc. Springer; 2006.

\bibitem{NavarroRychlik2007}
Navarro J, Rychlik T. Reliability and expectation bounds for coherent systems with exchangeable components. J. Multivar. Anal. 2007; 98:102--113.

\bibitem{Zhengcheng2010}
Zhengcheng Z. Ordering conditional general coherent systems with exchangeable components. J. Stat. Plan. Infer. 2010; 140:454--460.


\bibitem{Tavangar2014}
Tavangar M. Some comparisons of residual life of coherent systems with exchangeable components. Nav. Res. Logist. 2014;61:549-556. doi:10.1002/nav.21602


\bibitem{BP81}
Barlow RE, Proschan F. Statistical Theory of Reliability
and Life Testing, Madison, Silver Spring; 1981.


\bibitem{NavarroSandoval2007}
Navarro J, Ruiz JM, Sandoval CJ. Properties of Coherent Systems with Dependent Components. Commun. Stat.-Theory Methods 2007;36:175--191.

\bibitem{Duranteetal2007}
Durante F, Quesada-Molina JJ, \'{U}beda-Flores M. On a family of multivariate copulas for aggregation processes. Inf. Sci. 2007; 177:5715--5724.

\bibitem{Cuculescu2001}
Cuculescu I, Theodorescu, R. Copulas: Diagonals, tracjs. Rev. Roumaine Math. Pures Appl. 2001; 46:731--742.

\bibitem{Jaworski2009}
Jaworski P. On copulas and their diagonals. Inf. Sci. 2009; 179:2863--2871.


\bibitem{FredricksNelsen97}
Fredricks GA, Nelsen RB. Copulas Constructed from Diagonal Sections. Bene{\v{s}} V., {\v{S}}t{\v{e}}p{\'a}n J. (eds) Distributions with given Marginals and Moment Problems, 129--136. Springer Netherlands Dordrecht, 1997.


\end{thebibliography}
\end{document}